\documentclass[a4paper, 10pt, conference]{ieeeconf}      

\IEEEoverridecommandlockouts                              
\usepackage{graphics} 
\usepackage{epsfig} 
\usepackage{amsmath} 
\usepackage{amssymb}  
\usepackage{color}
\usepackage[algoruled]{algorithm2e}
\usepackage{cite}

\title{\LARGE 
Network Localization by Shadow Edges
}

\author{Gabriele Oliva$^*$, Stefano Panzieri$^\dagger$, Federica Pascucci$^\dagger$ and Roberto Setola$^*$\\
$^*$ University Campus Biomedico of Rome, Italy\\
g.oliva@unicampus.it r.setola@unicampus.it\\
$^\dagger$ Dipartimento di Informatica e Automazione, University ``Roma TRE'',\\
Via della Vasca Navale, 79, 00146, Roma, Italy.\\
panzieri@uniroma3.it pascucci@dia.uniroma3.it
}

\newtheorem{definition}{\textbf{Definition}}
\newtheorem{lemma}{\textbf{Lemma}}
\newtheorem{theorem}{\textbf{Theorem}}

\newtheorem{corollary}{\textbf{Corollary}}

\begin{document}

\maketitle
\thispagestyle{empty}
\pagestyle{empty}

\begin{abstract}
Localization is a fundamental task for sensor networks. 
Traditional network construction approaches allow to obtain localized networks requiring the nodes to be at least tri-connected (in 2D), i.e., the communication graph needs to be globally rigid. In this paper we exploit,  besides the information on the neighbors sensed by each robot/sensor, also the information about the lack of communication among nodes. The result is a framework where  the nodes are required to be bi-connected and the communication graph has to be rigid. This is possible considering a novel typology of link, namely Shadow Edges, that account for the lack of communication among nodes and allow to reduce the uncertainty associated to the position of the nodes.
\end{abstract}


\section{Introduction}
\label{sec:introduction}
Location service is a fundamental building block of many emerging computing/networking paradigms. 
Sensor location information is of the utmost importance for both sensor data integrity and important network management issues such as coverage and data delivery. Most of the approaches proposed in literature rely on Global Navigation Satellite System (GNSS) \cite{HoLi01}. However, GNSS is not suitable  for large-scale sensor localization due to its high cost, large form factor and, least but not last, environmental constraints. GNSS, indeed, requires direct line of sight to satellites. As a consequence, it does not work indoor, underground, and underwater. Moreover, due to the well known urban-canyon problem, GNSS provides poor accuracy in large metropolitan areas.

Recently, novel schemes have been proposed to determine the locations of the nodes in a network by exploiting some special nodes (called \textit{anchors}) which know their locations. According to this approach, localization algorithms derive the locations of sensor nodes from local measurements such as \textcolor{black}{relative} distance and angle estimations between neighbors \cite{ErGol93, GoKo04, MoLeo04}. 
Existing localization algorithms  \textcolor{black}{require} either the connectivity graph \cite{ShaRu03}, or the distances between neighboring sensor nodes \cite{SoYe05, ErGol93, GoKo04, MoLeo04} as input. One major challenge using this approach is localization ambiguity, since multiple, different localization solutions can satisfy all the distance constraints even if they are far from each other \cite{YuZhe11}. 

The problem of whether a graph with given edge length constraints admits a unique embedding in the plane is studied by rigidity theory \cite{La70}. A graph is rigid in the plane if one cannot continuously deform the shape of the graph without altering the lengths of the edges. A graph is globally rigid if it admits a unique embedding in the plane, subject to global rotations and translations. The theory of graph rigidity in 2D has been widely studied and well understood. 

In graph rigidity literature, many efforts have been made to explore the combinatorial conditions for rigidity. The Laman condition \cite{La70} characterizes graphs that are generically rigid. An efficient algorithm, the pebble game \cite{JaHe97}, is able to test whether a graph is generically rigid in time $O(nm)$ where $n$ is the number of nodes and $m$ is the number of edges. Similarly, for global rigidity, a sufficient and necessary condition \cite{JaJo05}  
based on the results in \cite{He92} \textcolor{black}{has been identified} by combining both redundant rigidity and $3$--connectivity. Both a combinatorial characterization of globally rigid graphs and polynomial algorithms for testing such graphs are \textcolor{black}{however,} not trivial to apply in the development of efficient localization algorithms. Given a graph with edge lengths specified, finding a valid graph realization in $\mathbb{R}^{d}$ for a fixed dimension $d$ is an NP-complete problem \cite{Sa79}.

Recently, Jackson and Jordan \cite{JaJo08} \textcolor{black}{proved} a sufficient condition based on $6$ mixed connectivity, which improves a previous result of $6$--connectivity by \cite{LoYe82}. There are also some results for random geometric graphs. Assuming the unit disk model, many researchers \cite{GuKu98, LiWa99, Pe99, XuKu04} considered critical conditions for graph connectivity. Simulation results \cite{ErGol93} ensure that the hitting radius of global rigidity is between $3$-- and $6$--connectivity in \textcolor{black}{a probabilistic} sense.

To the best of our knowledge, all the approaches proposed in literature partially exploit the topology information. They consider as good, indeed, only the information of \textit{being connected} and completely discarded the one of \textit{not being connected}. 
\textcolor{black}{This papers shows how to exploit connectivity graph taking into account both the information on connectivity (i.e., the presence of an edge between two nodes) and lack of connectivity (i.e., the absence of a link).}
Specifically the \textit{being connected} relationship carries information about where a node is located (\textcolor{black}{i.e.}, \textit{admissible regions}), while the \textit{not being connected} one carried information about where a node cannot be located (\textcolor{black}{i.e.,} \textit{forbidden regions}). 
In this paper the problem of finding a unique realization is addressed by exploiting both the admissible and forbidden region actively. According to this approach the conventional rigidity requirement is no longer a necessary condition to find a unique realization. To this end, rigid graphs become only a subset of the total set of networks where a unique solution can be found. To demonstrate this, the notions of \textit{shadow edges} and \textit{shadow graph} are introduced and the realization problem is studied considering a combination of both \textcolor{black}{regular and shadow edges}. 
Furthermore, an algorithm to build localized network having polynomial complexity is proposed.

The paper is organized as follows: Section \ref{sec:localization_and_rigidity} provides a general overview on localization and rigid graphs; the proposed approach is detailed in Section \ref{sec:localization_and_shadow};
some simulation results and comparison with standard localized network construction algorithm are provided in Section \ref{sec:simulation_results}; finally, some remarks and future work directions are collected in Section \ref{sec:conclusions}.

\section{Localization and Rigidity}
\label{sec:localization_and_rigidity}
In this section, the network localization problem with distance information is addressed: as explained above, it can be studied in the framework of graph theory.

To this end, we need some initial definitions.

\begin{definition}
Let $\mathcal{G}=\{\mathcal{V},\mathcal{E}\}$ be a {\em graph} with $n$ nodes, where  the set $\mathcal{V}$ denotes the nodes ${\it v}_1, \ldots, {\it v}_p$; $\mathcal{E}$ is the set of edges $({\it v}_i,{\it v}_j)$.
The graph can be also represented by means of an {\em adjacency matrix} ${ \Gamma}=\{\gamma_{ij}\}$ that is composed of non-negative entries  $\gamma_{ij} =1 \Leftrightarrow ({\it v}_i,{\it v}_j) \in \mathcal{E} $,
i.e., there exists an arc that starts form node $v_i$ and reaches node $v_j$ (note that it is also possible to consider non-unitary weights). 

The graph $\mathcal{G}$ is said to be \emph{undirected} if $({\it v}_j,{\it v}_i)\in \mathcal{E}$ whenever $({\it v}_i,{\it v}_j)\in \mathcal{E}$ (the weights are $\gamma_{ij}=\gamma_{ji}$ in this case); otherwise the graph $\mathcal{G}$ is said to be \emph{directed}.
\end{definition}

\begin{definition}
Let a set of $n$ sensors $\Sigma=\{\sigma_i\}$ embedded in $\mathbb{R}^d$, and let $d_{ij}$ be the  distances  between sensors $\sigma_{i}$ and  $\sigma_{j}$. Suppose that the coordinates $p_i \in \mathbb{R}^d$ of small number of sensors $\sigma_j\in \Sigma_j \subseteq \Sigma$ are known, and denote $\Sigma_j$ as the \textit{kernel sensor set}.
Each sensor $\sigma_i$is assumed to be located at a fixed position in $\mathbb{R}^{d}$ and is associated with it a specific set of neighboring sensors. In this paper only symmetric neighboring relation is considered, thus a sensor $\sigma_{j}$ is a neighbor of a sensor  $\sigma_{i}$ if and only if sensor  $\sigma_{i}$ is also a neighbor of sensor $\sigma_{j}$.
 The \emph{localization problem} consists in finding a map $Q : \Sigma \rightarrow \mathbb{R}^{d}$ (where $d$ is $2$ or $3$) which assigns coordinates $p_{i}\in \mathbb{R}^{d}$ to each sensor $\sigma_{i}$ such that $||p_{i}-p_{j}|| = d_{ij}$ holds for all pairs $(i, j)$ for which $d_{ij}$ is given, and the assignment is consistent with any sensor coordinate assignments provided in the problem statement.
\end{definition}

Under these conditions, a graph $\mathcal{G}=\{\mathcal{V},\mathcal{E}\}$ can be correlated with with a sensor network by associating a vertex $v_i$ of the graph with each sensor $\sigma_i$, and an edge of the graph $({\it v}_i,{\it v}_j)$ with each sensor pair $\sigma_i,\sigma_j$ for which the inter-sensor distance is known.
\begin{definition}
A \emph{$d$--dimensional framework} $(\mathcal{G}, Q)$ is a graph $\mathcal{G}=\{\mathcal{V},\mathcal{E}\}$ together with a map  $Q : \mathcal{G}\rightarrow \mathbb{R}^{d}$. The framework is a \emph{realization} if it results in $||p_{i}-p_{j}|| = d_{ij}$  for all pairs $(\sigma_i, \sigma_j)$ such that $(v_i, v_j)\in \mathcal{E}$. According with this approach, the localization over $\Sigma$ problem is mapped into a realization problem over $(\mathcal{G}, Q)$.
\end{definition}

\begin{definition}
The network localization problem just formulated is said to be {\em solvable} if there is exactly one set of vectors $\{ p_{1}, p_{2}, \dots p_{n}\}$ is consistent with the given data $(\mathcal{G},Q)$. 
\end{definition}

To understand the solvability of the localization problem, let us consider the notions of rigidity and global rigidity.

\begin{definition}
Two frameworks $(\mathcal{G},Q)$ and $(\mathcal{G},Q^*)$ are \textit{equivalent} if $||p_{i}-p_{j}|| = ||p^*_{i}-p^*_{j}||$ holds for all pairs $(i, j)$ with $(v_i, v_j)\in \mathcal{E}$. 
\end{definition}

\begin{definition}
Two frameworks $(\mathcal{G},Q)$ and $(\mathcal{G},Q^*)$ are \textit{congruent} if if $||p_{i}-p_{j}|| = ||p^*_{i}-p^*_{j}||$ holds for all pairs $i, j$ with $v_i, v_j \in \mathcal{V}$. 
\end{definition}

This is the same as saying that $(\mathcal{G},Q)$ can be obtained from $(\mathcal{G},Q^*)$ by an isometry over $\mathbb{R}^{d}$, i.e., a combination of translations, rotations and reflections.

\begin{definition}
A framework $(\mathcal{G},Q)$ is \textit{rigid} if there exists a sufficiently small positive $\epsilon$ such that if $(\mathcal{G},Q)$ is equivalent to $(\mathcal{G},Q^*)$ and $||p_i - q_i||< \epsilon$ for all $v_i \in \mathcal{V}$ then $(\mathcal{G},Q)$ is congruent to $(\mathcal{G},Q^*)$. 
\end{definition}

Intuitively, a rigid framework cannot flex. 
\begin{definition}
A framework $(\mathcal{G},Q)$ is \textit{globally rigid} if every framework which is equivalent to $(\mathcal{G},Q)$ is congruent to $(\mathcal{G},Q)$. 
\end{definition}

Obviously, if $\mathcal{G}$ is a complete graph then the framework $(\mathcal{G},Q)$ is necessarily globally rigid.

According with this approach, a localization problem can be solved only if the graph framework $(\mathcal{G},Q)$ is globally rigid (i.e., otherwise the position of the sensors would not be univocally determined, since flipping some node would result in another admissible solution). 

Although the realization of general globally rigid graphs is hard,  a class of globally rigid graphs can be computationally efficient to realize by means of {\em trilateration}. Trilateration is the operation whereby a node with known distances from three $(d = 2)$ or four $(d = 3)$ other nodes, determines its own position in terms of the positions of those neighbors. An easy way to obtain a localizable network (i.e., globally rigid graph) is by iteratively adding nodes attached to al least three (four) nodes. This in the following will be referred to as \emph{Trilateration Network  Construction} (TNC) algorithm.

In the following, we demonstrate that considering also information about non-neighboring nodes, the global rigidity assumption can be relaxed and a different and more effective procedure to build localizable network can be set up.

It is worth to underline that, given the graph and distance set of a globally rigid framework, there is not enough information to position the framework absolutely in $\mathbb{R}^{d}$. To do this requires the absolute position of at least three $(d = 2)$ or four vertices $(d = 3)$ non-collinear vertices, i.e., \textit{kernel} nodes.

\section{Localization and Shadow Edges}
\label{sec:localization_and_shadow}
 In the following we will assume to generate a network iteratively, by adding each node to a number of existing nodes.
 Specifically, each node $i$ is provided with a maximum communication radius $\rho$ (assumed to be the same for all the agents) and is able to detect the presence of each and every node $j$ that falls within such communication radius, obtaining also an information on the distance $d_{ij}\leq \rho$ among them (we assume $d_{ij}=d_{ji}$).

\begin{definition}
The set of {\em localization options} $\mathcal{L}_i(\mathcal{G})$ for a node $v_i$ in a graph $\mathcal{G}$, is the set of points $[x,y]^T \in \mathbb{R}^2$ that are admissible for the position of node $v_i$, given the structure of the graph $\mathcal{G}$.  
 \end{definition}
 
  \begin{definition}
 A node $v_i$ is {\em localized} over a graph $\mathcal{G}$ provided that its position is univocally determined. 
 Otherwise the node is not localized.
 \end{definition}

 \begin{definition}
 A graph $\mathcal{G}$ is {\em localized} provided that each node $v_i \in \mathcal{V}$ is localized.
 \end{definition}

Note that a node $v_i$, being able to sense  its neighbors $v_j$ within the communication range $\rho$,
is able to obtain the distance $d_{ij}$; moreover, if  the nodes $v_j$ are localized, their position is known. 
Therefore, assuming the graph $\mathcal{G}$ to be connected the following options are possible, depending on the cardinality $|\mathcal{L}_i(\mathcal{G})|$:
\begin{itemize}
\item $|\mathcal{L}_i(\mathcal{G})|=1$ if and only if node $v_i$ is connected to at least $3$ non collinear localized nodes;
\item $|\mathcal{L}_i(\mathcal{G})|=2$ if and only if node $v_i$ is connected to $2$ localized nodes;
\item $|\mathcal{L}_i(\mathcal{G})|=\infty$ if node $v_i$ is connected to $1$ localized node
\end{itemize}
 
 \begin{definition}
A {\em sensing disk} $S_{i,h}$ for a node $v_i$ is a disk or radius $\rho$ centered in a localization option $[x_h,y_h]^T \in \mathcal{L}_i(\mathcal{G})$, i.e., 
\begin{equation}
S_{i,h}=\{[x,y]^T \in \mathbb{R}^2: (x-x_h)^2+(y-y_h)^2\leq \rho^2,\}.
\end{equation}

 \end{definition}

\begin{definition}
The {\em admissible sensing region} for a node $v_i$, denoted by $\mathcal{D}_i$ is the union of the sensing disks $S_{i,h}$ for each of the localization options $[x_h,y_h]^T \in \mathcal{L}_i(\mathcal{G})$.
\end{definition}

As discussed above, if each node is connected to at least $3$ non collinear localized nodes, the network is localized.
However it will be shown that, if a node $v_i$ is such that $|\mathcal{L}_i(\mathcal{G})|=2$ (e.g., it is connected to $2$ localized nodes), it may still be possible to localize it.
To this end we need to define the following edge class.

 \begin{definition}
A {\em shadow edge} is an edge that connects two nodes $v_i$ and $v_j$ if and only if one of the following holds true:
\begin{itemize}
\item $v_i$ is localized over $\mathcal{G}$ and $|\mathcal{L}_j(\mathcal{G})|=2$; $\mathcal{L}_i(\mathcal{G}) \in \mathcal{D}_j$ but $d_{ij}>\rho$;
\item $v_j$ is localized over $\mathcal{G}$ and $|\mathcal{L}_i(\mathcal{G})|=2$; $\mathcal{L}_j(\mathcal{G}) \in \mathcal{D}_i$ but $d_{ji}>\rho$.
 \end{itemize}
 \end{definition}

Such an edge can be created when node $v_i$ is not detected by $v_j$ but $v_i\in \mathcal{D}_j$ or vice versa. Hence, it represents the fact that node $v_i$ has $2$ localization options and one of them should not be taken into account; in fact if node $v_i$ was in the localization option such that $v_j$ lies in the corresponding disk, it should have sensed node $v_j$. Therefore node $v_i$ it is in the other localization option.
Note that the node $v_j$ cannot lie in the intersection of the $2$ disks, otherwise it would have been connected to $v_i$ with a regular edge.
Figure \ref{fig:esempio_4_nodi} shows an example of localization by means of a shadow edge.

\begin{figure}[!ht]
\begin{center}
\includegraphics[width=3.5in]{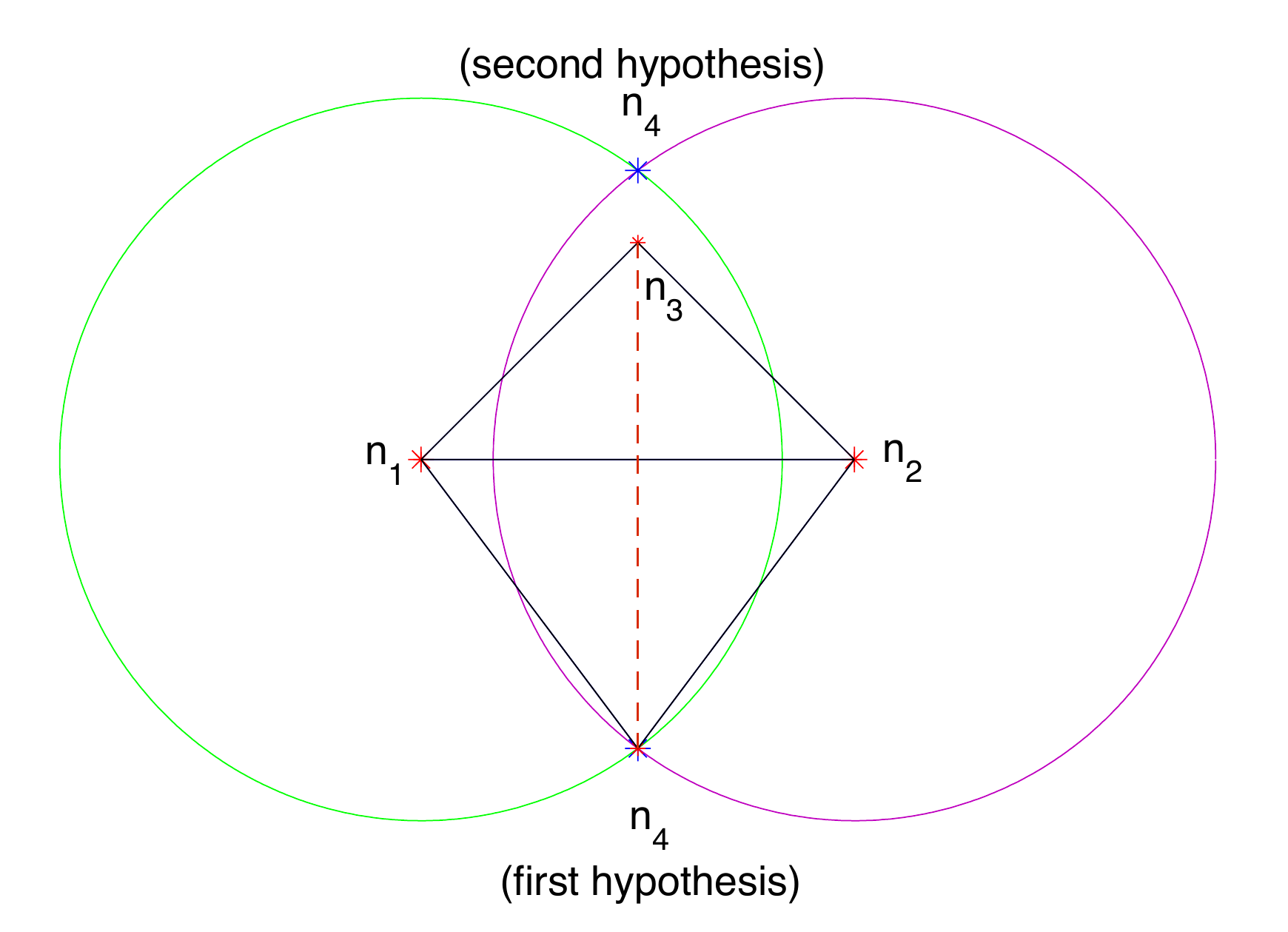}
\caption{Example of shadow edge: a new node $n_4$ attempts to connect to the network. Since the node $n_4$ senses nodes $n_1$ and $n_2$ two hypotheses are possible for its position (i.e., the intersection of the green and purple circumferences). Since node $n_3$ is not detected by $n_4$, it is possible to exclude hypothesis $2$ and node $n_4$ is localized by adding a shadow edge (red dotted line) between $n_4$ (first hypothesis) and $n_3$.}
 \label{fig:esempio_4_nodi}
\end{center}
\end{figure}

 \begin{definition}
A {\em shadow graph} is the graph $$\mathcal{G}_s=\{\mathcal{V},\mathcal{E}_s\}$$ where $\mathcal{V}$ is the set of nodes of graph $\mathcal{}G$ and $\mathcal{E}_s$ is the set of the shadow edges.
 \end{definition}
 
  \begin{definition}
A {\em extended shadow  graph} is the graph $$\mathcal{G}_e=\{\mathcal{V},\mathcal{E}\cup \mathcal{E}_s\}$$ where $\mathcal{G}=\{\mathcal{V},\mathcal{E}\}$ is the nominal graph and $\mathcal{E}_s$ is the set of the shadow edges.
 \end{definition}
 In the following we will assume, for the sake of simplicity and without loss of generality, that the event of connecting a node $v_i$ to three collinear nodes $v_j,v_k,v_h$ either in $\mathcal{G}$ and in $\mathcal{G}_e$ is never verified.
   \begin{definition}
A {\em minimal shadow  graph} is the graph $$\mathcal{G}_s=\{\mathcal{V},\mathcal{E}^*_s\}$$ where $\mathcal{G}=\{\mathcal{V},\mathcal{E}\}$ is the nominal graph and $\mathcal{E}^*_s$ is a {\em minimal set of shadow edges} defined as follows:
$$
\forall i:  |\mathcal{L}_i(\mathcal{G})|=2 \quad \exists ! \mbox{ shadow edge } (v_i, v_j)\in \mathcal{E}^*_s.
$$

 \end{definition}
Let us now provide some results on localization using such shadow edges.

\begin{lemma}
Suppose that a graph $\mathcal{G}$ is connected and that $\mathcal{E}_s$ is a set of shadow edges. Then the nodes are localized over $\mathcal{G}_e$ if and only if the extended shadow graph $\mathcal{G}_e$ is globally rigid.
\end{lemma}
\begin{proof}
Clearly, if $\mathcal{G}_e$ is globally rigid, then $|\mathcal{L}_i(\mathcal{G}_e)|=1$ for each node $i=1, \ldots, n$ and the nodes are  localized. 
To prove the reversed implication, suppose that the nodes are all localized over $\mathcal{G}_e$ but $\mathcal{G}_e$ is not globally rigid: then there is at least a node $i$ such that $|\mathcal{L}_i(\mathcal{G})|\geq 2$, but this is not possible, since the network is localized, yielding to an absurd.
%
%
%
\end{proof}

The above Lemma, therefore, allows to infer the localization of the nodes in $\mathcal{G}$ by considering the union of the regular edges and of the shadow edges. According to such an approach, let us provide the following Theorem.

\begin{theorem}
Suppose that the graph $\mathcal{G}$ is connected and that $\mathcal{E}_s$ is a set of shadow edges; then the following propositions are equivalent:
\begin{enumerate}
\item the nodes are localized over $\mathcal{G}_e$;
\item $\max_{i=1,\ldots,n}|\mathcal{L}_i(\mathcal{G})|\leq2$ and each node $v_i\in \mathcal{G}$ such that $|\mathcal{L}_i(\mathcal{G})|=2$ is connected to at least $2$ localized nodes and there is at least a  shadow edge $(v_i,v_j)$ with $v_j \in \mathcal{V}$

\item $\mathcal{G}$ is locally rigid and there is a minimal shadow edge set $\mathcal{E}^*_s \in \mathcal{E}_s$.

\end{enumerate}
\end{theorem}
\begin{proof}
Let us first prove  $1)\Leftrightarrow 2)$. 
Note that, if the assumptions in $2)$ are verified, then each node in $\mathcal{G}_e$ is at least tri-connected and the nodes are localized. Conversely, assume that the nodes are not localized over $\mathcal{G}_e$ and that the hypotheses in $2)$ hold: then there is at least a node $v_i$ that is not localized. This implies $|\mathcal{L}_i(\mathcal{G}_e)|\geq 2$, which is contrast with the hypotheses in $2)$.
Let us now prove $2)\Leftrightarrow 3)$. Note that, if $m$ is the number of nodes such that $|\mathcal{L}_i(\mathcal{G})|=2$, then choosing at least a shadow edge for each of such nodes implies $|\mathcal{E}|\geq m$.
Note that a minimal shadow edges set is such that 
$$
|\mathcal{E}^*_s|=\sum_{i=1}^p |\mathcal{L}_i|-n=m
$$
hence $\mathcal{E}^*_s \subseteq \mathcal{E}_s$.
Conversely, choosing a minimal set of shadow edges  implies $2)$ if  $\max_{i=1,\ldots,n}|\mathcal{L}_i(\mathcal{G})|\leq2$.
Note that if $\mathcal{G}$ is not locally rigid, then there is at least a node $v_i$ with  $|\mathcal{L}_i(\mathcal{G})|=\infty$; conversely if there is at least a node $v_i$ with  $|\mathcal{L}_i(\mathcal{G})|=\infty$ the graph $\mathcal{G}$ is not locally rigid.

\end{proof}

The above Theorem provides a useful algorithm for checking if a given graph $\mathcal{G}$ is localized:

%

\begin{algorithm}
\SetKw{Set}{{\bf Set}}
\caption{Graph Localization Check}
\label{AN_algo}
\BlankLine
\KwData{Undirected graph $\mathcal{G}=(\mathcal{V},\mathcal{E})$}
\KwResult{success or fail}
\BlankLine
\For{i=1 \emph{\KwTo} n}{
	\eIf{$|\mathcal{L}_i(\mathcal{G})|\geq 3$}{
	\Return fail\;
	}{
	\Set $c_j=\exists (v_i,v_j)\in \mathcal{E}: d_{ij}\leq \rho \wedge |\mathcal{L}_j(\mathcal{G})|=1$\;
		\Set $c_k=\exists (v_i,v_k)\in \mathcal{E}:
		v_k\neq v_j \wedge d_{ik}\leq \rho \wedge |\mathcal{L}_k(\mathcal{G})|=1$\;
		\Set $c_h=\exists v_k\in \mathcal{V}:$
		$d_{ik}> \rho \wedge |\mathcal{L}_h(\mathcal{G})|=1\wedge \mathcal{L}_h(\mathcal{G})\subset \mathcal{D}_i$\;
		\If{not($c_j \wedge c_k \wedge c_h$)}{
			\Return fail\;
		}
	}
}
\Return success\;
\end{algorithm}

Let us now provide a criterion for the iterative construction of a localized network that does not require the graph to be rigid.
\begin{corollary}
Consider a growing network starting with a complete graph of $3$ localized nodes. then if each new node $v_i$ is such that it is connected to $2$ non collinear preexisting nodes $v_j,v_k$ such that $|\mathcal{L}_j(\mathcal{G})|=|\mathcal{L}_k(\mathcal{G})|=1$ and it is possible to find at least a shadow edge $(v_i,v_k)$, then the resulting network is localized.
\end{corollary}
\begin{proof}
Note that $|\mathcal{L}_i(\mathcal{G})|\leq 2$ for all $v_i\in \mathcal{V}$ and there is a shadow edge for each node $v_i$ such that $|\mathcal{L}_i(\mathcal{G})|=2$; therefore, according to Theorem $1$,  the nodes are localized.
\end{proof}

Let us now provide an algorithm that iteratively constructs a network according to Corollary 1.
The algorithm, starting with a core of $3$ localized nodes, generates nodes in random position.
However such an information is used only to detect the nodes and the distances, while the assessment of the position is done by using the shadow edges.

\begin{algorithm}
\SetKw{Set}{{\bf Set}}
\SetKw{Trilaterate}{{\bf Trilaterate}}
\SetKw{Return}{{\bf Return}}
\caption{Localized Network Construction}
\label{AN_algo}
\BlankLine
\KwData{Graph $\mathcal{G}=(\mathcal{V},\mathcal{E})$ with $3$ connected and localized nodes. Target size $n>3$ of the graph.}
\KwResult{Final graph $\mathcal{G}=(\mathcal{V},\mathcal{E})$; position $[x_i,y_i]^T$ for each node $v_i\in\mathcal{V}$}
\BlankLine
\tcc{Initialisation}
\Set $\mathcal{E}_s=\emptyset$;\\
\tcc{Main Cycle}
\For{i=4 \emph{\KwTo} n}{
        		Generate $v_i$ in random {\em unknown} position\;
		\Set $c_j = \exists v_j\in \mathcal{V} : $\\ $|\mathcal{L}_j(\mathcal{G})|=1 \wedge d_{ij}\leq \rho$\;
		\Set $c_k = \exists v_k \in \mathcal{V}-\{v_j\}: $
		$ |\mathcal{L}_k(\mathcal{G})|=1 \wedge  d_{ik}\leq \rho$\;
		\If {$c_j \wedge  c_k$}
		{
			\Set $\mathcal{V}=\mathcal{V}\cup \{v_i\}$\;
			\Set $\mathcal{E}_{tmp}=\{(v_i,v_j): d_{ij}\leq \rho\}$\;
			\eIf{$|\mathcal{E}_{tmp}|\geq 3$}{
				\tcc{Standard Trilateration}
				\Set $\mathcal{E}=\mathcal{E}\cup \mathcal{E}_{tmp}$\;
				\Trilaterate $[x_i,y_i]^T$ over $(\mathcal{V}, \mathcal{E})$\;
			}
			{
				\tcc{Shadow Edges Trilateration}
				\Set $c_h=\exists v_h : 
				d_{ih}> \rho \wedge  |\mathcal{L}_h(\mathcal{G})|=1 \wedge [x_h,y_h]^T\in \mathcal{D}_i$\;
			}
			\If {$|\mathcal{E}_{tmp}|=2 \wedge  c_h$}
			{
				\Set $\mathcal{E}=\mathcal{E}\cup \mathcal{E}_{tmp}$\;
				\Set $\mathcal{E}_s=\mathcal{E}_s\cup \{(v_i,v_h)\}$\;
				\Trilaterate $[x_i,y_i]^T$   over $(\mathcal{V},\mathcal{E}\cup \mathcal{E}_s)$\;
			}
		}
}
\end{algorithm}

\section{Simulation Results}
\label{sec:simulation_results}

Let us now provide some simulation results and a comparison between the proposed methodology and the TNC algorithm for the construction of localized graphs.
In the following we will consider networks where the nodes are embedded in the unit square in $\mathbb{R}^2$, i.e., for each node $i$, $x_i,y_i\in[0,1]$.

Let us first provide an example of execution of the proposed algorithm. Figure \ref{fig:grafo40} shows an example of graph with $40$ nodes. The nodes with blue incident edges are localized according to both the proposed algorithm and the standard localization algorithm. The nodes with incident shadow edges (purple edges) are localized according to the proposed algorithm; note that in the figure all possible shadow edges are shown. The nodes with green incident edges are not localized. The proposed algorithm localizes $29$ nodes, the standard algorithm localizes $18$ nodes.

Let us now provide some results on the comparison between the proposed algorithm and the TNC algorithm.
Figures  \ref{fig:1_shadow_edge_loc}, \ref{fig:3_difference} and \ref{fig:4_percentage} show a comparison with the TNC approach without shadow edges; results are the average of $50$ runs.
Figure \ref{fig:1_shadow_edge_loc} shows the percentage of localized nodes obtained by Algorithm 2 with respect to total generated nodes, plotted against the distance radius $\rho$ and the network size.
As shown by the Figure, and just as expected, for a fixed network size there is a threshold on the distance radius, and the nodes are almost all localized for valued of $\rho$ above such a threshold. Note further that such a threshold value tends to decrease when the network size increases.

\begin{figure}[!ht]
\begin{center}
\includegraphics[width=3.4in]{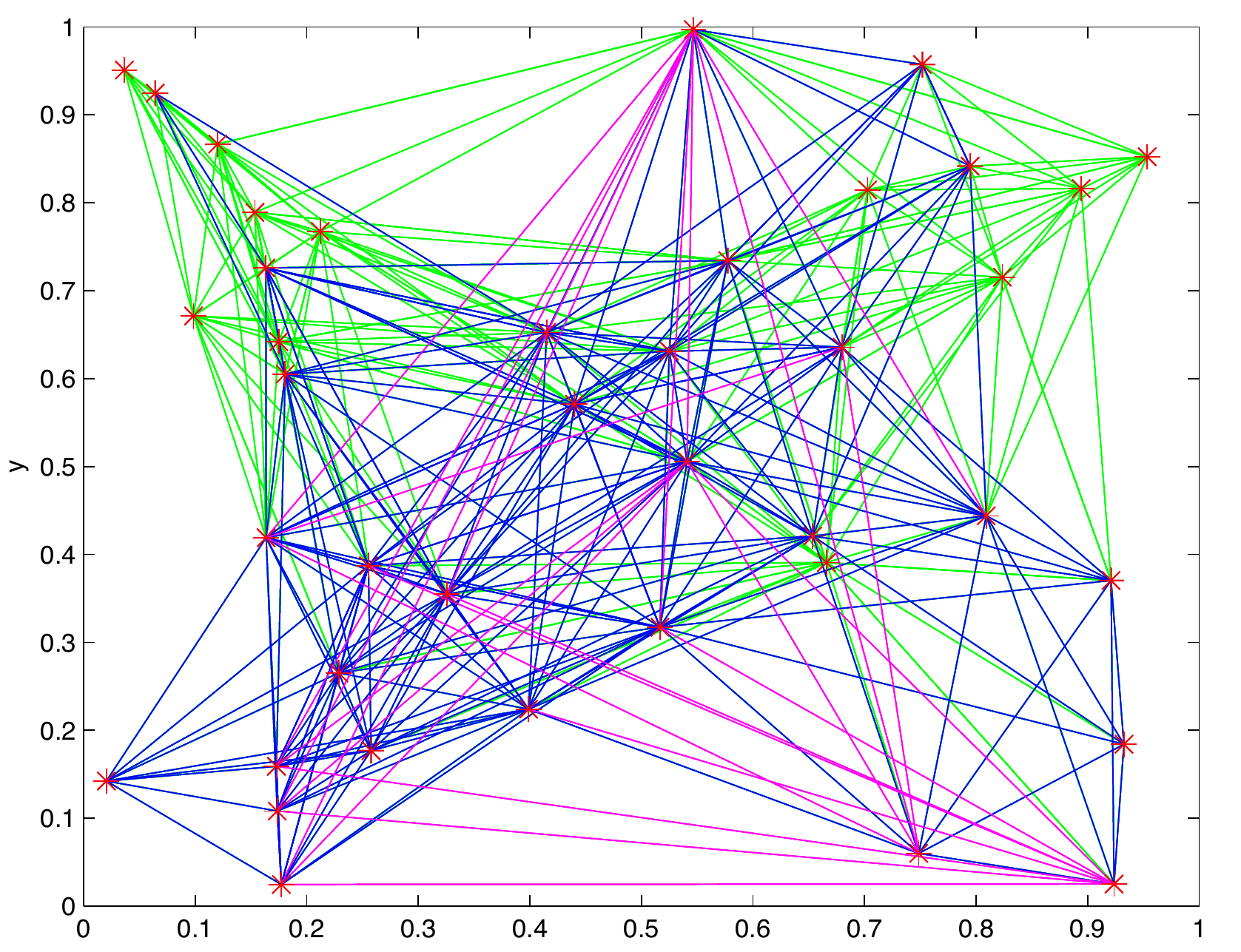}
\caption{Example of graph with $40$ nodes. The nodes with blue incident edges are localized according to both the proposed algorithm and the TNC algorithm. The nodes with incident shadow edges (purple edges) are localized according to the proposed algorithm. The nodes with green incident edges are not localized. The proposed algorithm localizes $29$ nodes, the TNC algorithm localizes $18$ nodes.}
\label{fig:grafo40}
\end{center}
\end{figure}

Figures \ref{fig:3_difference} and \ref{fig:4_percentage} show, respectively, the difference and the ratio of the percentage of localized nodes obtained by Algorithm 2 and by the TNC algorithm.
Both Figures highlight that the proposed algorithm is particularly effective for values of $\rho$ below the network size-dependent  threshold values identified in Figure \ref{fig:1_shadow_edge_loc}. 
Specifically, note that the percentage of the proposed algorithm are always greater or equal  than the TNC algorithm; moreover there are up to $25\%$ more localized nodes (Figure \ref{fig:3_difference}), that is about twice the percentage of localized nodes of the TNC algorithm (Figure  \ref{fig:4_percentage}).
Figure \ref{fig:percshadowedge} shows that, the range of values of $\rho$ and $N$ where the proposed algorithm enhances the localization of nodes, is indeed the range where the shadow edges are created.
Note that, in such a range, the amount of shadow edges is about $10-12 \%$ of the total number of edges created.

\begin{figure}[!ht]
\begin{center}
\includegraphics[width=3.5in]{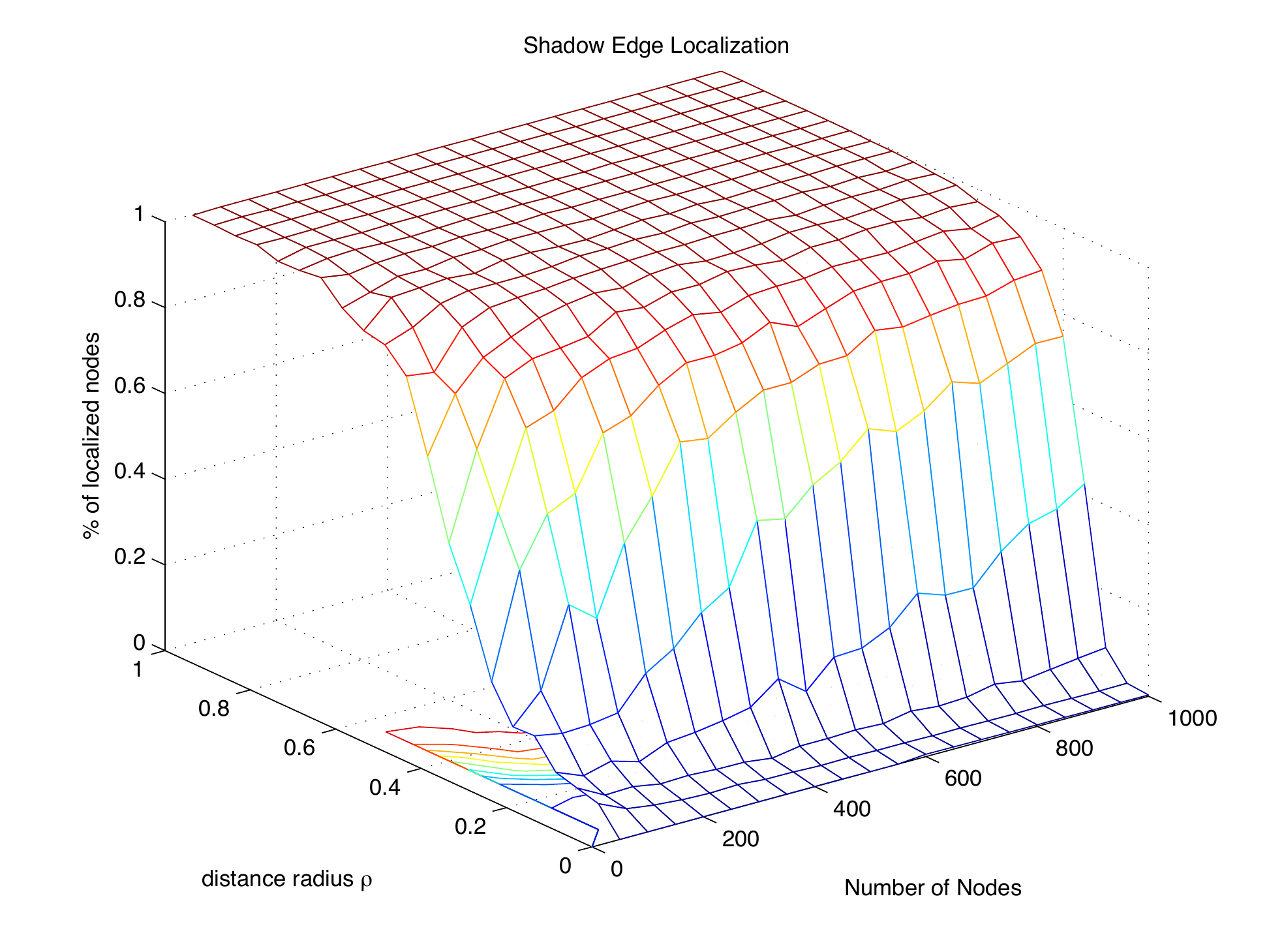}
\caption{Percentage of localized nodes obtained by Algorithm 2 with respect to total generated nodes, plotted against the distance radius $\rho$ and the network size. Results are the average of $50$ runs.}
\label{fig:1_shadow_edge_loc}
\end{center}
\end{figure}

\begin{figure}[!ht]
\begin{center}
\includegraphics[width=3.5in]{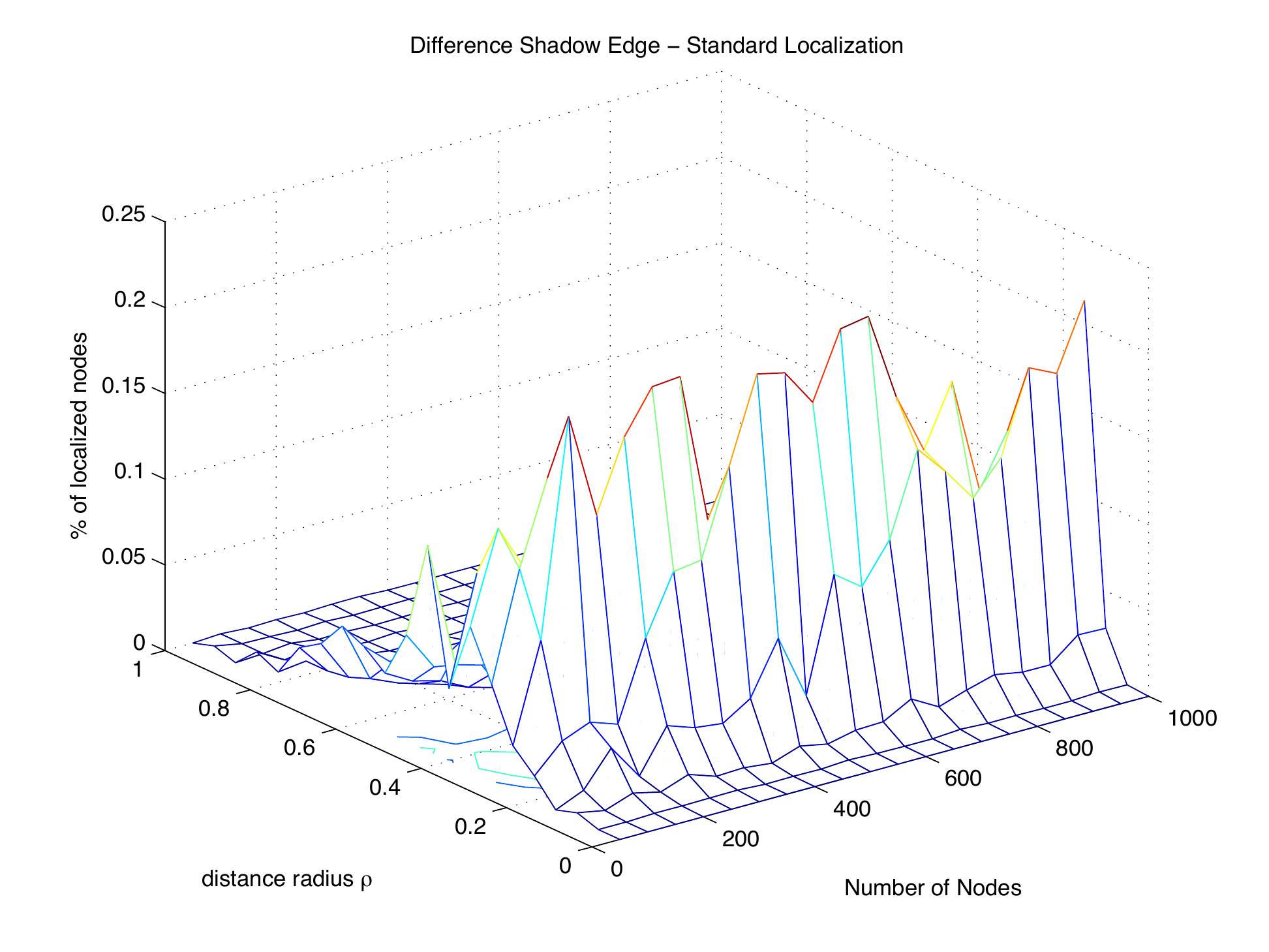}
\caption{Difference between the percentage of localized nodes with respect to total generated nodes, obtained by Algorithm 2 and by the TNC algorithm, plotted against the distance radius $\rho$ (the same for all nodes) and the network size. Results are the average of $50$ runs.}
\label{fig:3_difference}
\end{center}
\end{figure}

\begin{figure}[!ht]
\begin{center}
\includegraphics[width=3.5in]{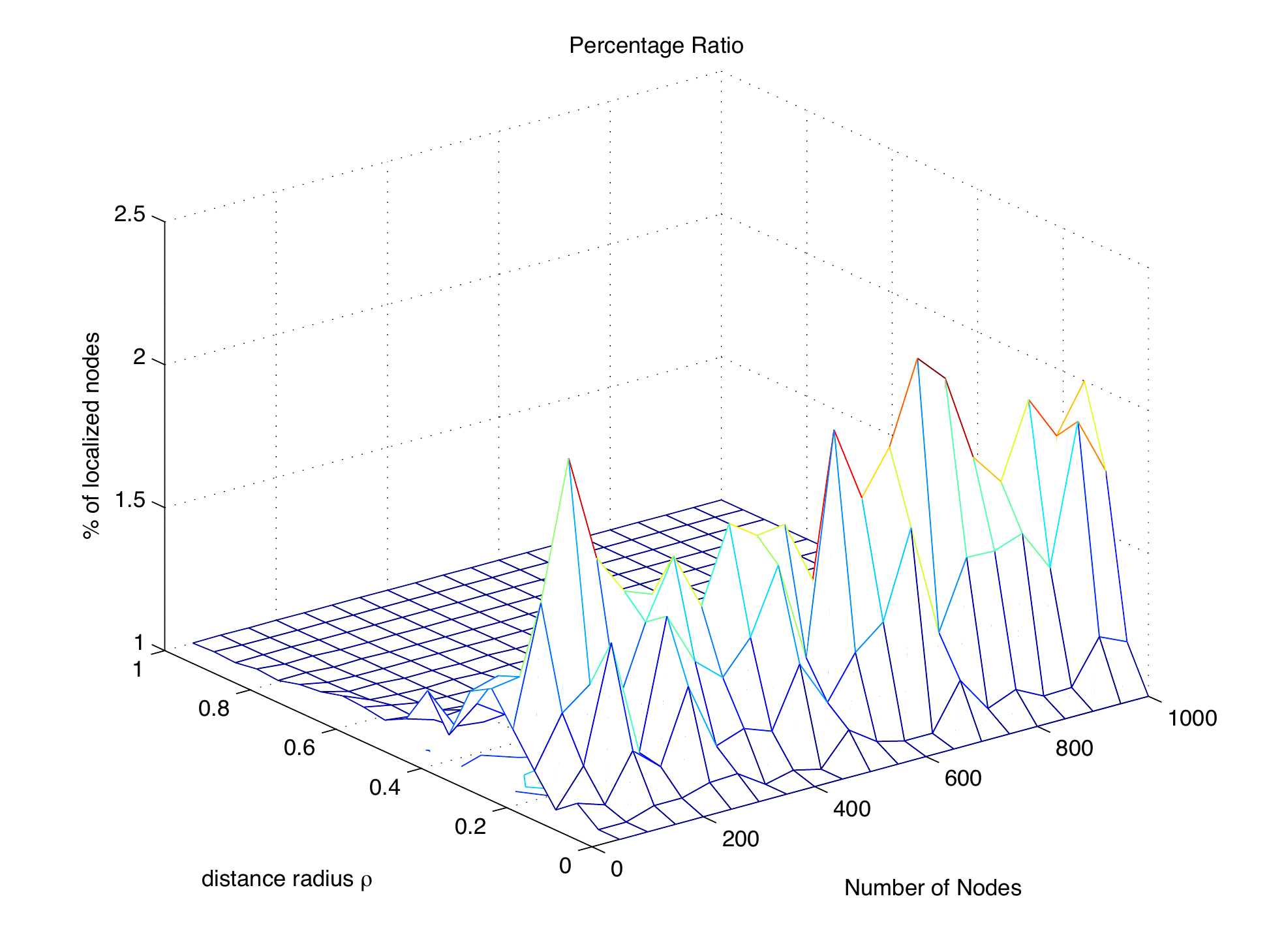}
\caption{Ratio between the percentage of localized nodes with respect to total generated nodes, obtained by Algorithm 2 and by the TNC algorithm, plotted against the distance radius $\rho$ and the network size. Results are the average of $50$ runs.}
\label{fig:4_percentage}
\end{center}
\end{figure}

\begin{figure}[!ht]
\begin{center}
\includegraphics[width=3.5in]{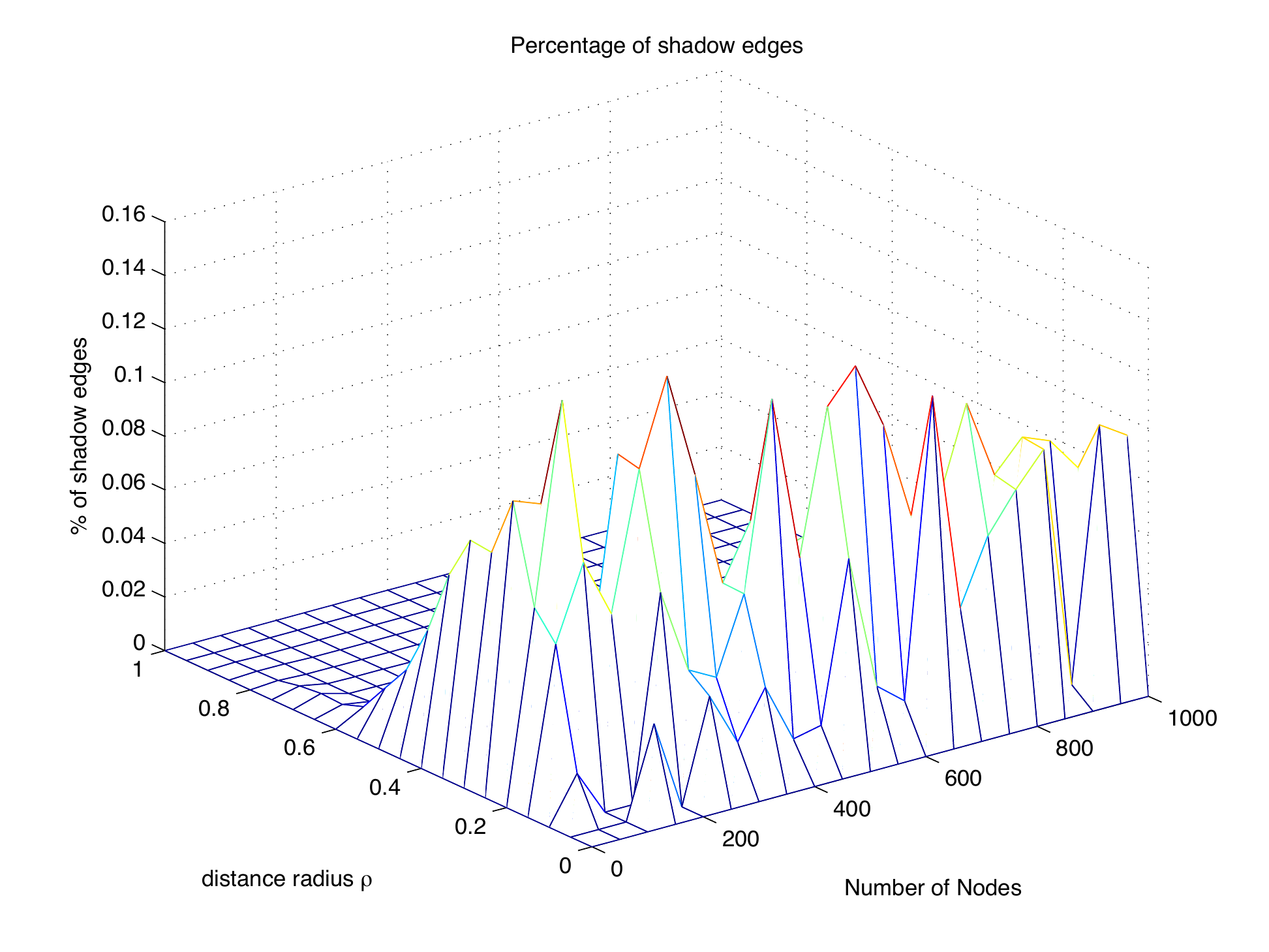}
\caption{Number of shadow edges created by the algorithm, plotted against the distance radius $\rho$  and the network size. Results are the average of $50$ runs.}
\label{fig:percshadowedge}
\end{center}
\end{figure}

\section{Conclusions} 
\label{sec:conclusions}
In this paper a methodology for the construction of localized networks is proposed which is less demanding than the TNC algorithm. The idea is to use the information about the lack of connectivity between nodes in order to enhance the localization procedure, thus requiring (in 2D) each node to sense just $2$ neighbors nodes, instead of the 3 required by the TNC algorithm.

Future work will be aimed to extend the methodology in order to provide a reasonably small set of localization options for the nodes, in the case where only a single real link is available, by considering multiple shadow edges.
The possibility to localize, at least to a certain extent, a network made in such a way would allow agents to construct radio bridges (e.g., chains) for the ad-hoc communication in harsh environments (e.g., robotic swarms operating in tunnels or corridors during a fire blast). 
Moreover the possibility to distribute the localization procedure will be inspected.

%
%
%
%
%
%

\bibliographystyle{IEEEtran}

\end{document}